\documentclass[12pt,twoside,amsfonts,amssymb]{article}
\usepackage{amsthm}
\usepackage{amsmath}
\usepackage{amsfonts}
\usepackage{amssymb}
\usepackage{epsf}
\usepackage{graphicx}
\usepackage{revsymb}
\usepackage[dvips]{color}
\usepackage{amscd}
\usepackage{times}
\usepackage{titletoc}
\usepackage[small,center,compact,outermarks]{titlesec}
\usepackage{a4}
\DeclareFontFamily{OT1}{rsfs}{}
\DeclareFontShape{OT1}{rsfs}{m}{n}{ <-7> rsfs5 <7-10> rsfs7 <10-> rsfs10}{}
\DeclareMathAlphabet\mathcurl{OT1}{rsfs}{m}{n}
\contentsmargin{0pt}

\titleformat{name=\chapter}[display]{\normalfont\Large\filcenter}
{\titlerule[1pt]\vspace{1pt}\titlerule\vspace{1ex}\Large\chaptername\ \thechapter}{1ex}
{\LARGE\bfseries}[\vspace{1ex}\titlerule]
\titleformat{name=\chapter,numberless}[display]{\normalfont\Large\filcenter}
{\titlerule[1pt]\vspace{1pt}\titlerule}{-2ex}{\LARGE\bfseries}[\vspace{1ex}\titlerule]
\titlespacing{\chapter}{0pt}{-2em}{2em}
\titlecontents{chapter}[0pc]
{\addvspace{1.5pc}%
\filright}
{\small \bfseries\sffamily Chapter~\thecontentslabel. \ }{\small\bfseries\sffamily}
{\hfill \quad\small\bfseries\sffamily\thecontentspage}
\titleformat{name=\section}[block]{\bfseries\filright}{\thesection}{1em}{}
\titlespacing{\section}{0pt}{3em}{1em}
\titlecontents{section}[0pc]
{\addvspace{1.0pc}%
\filright}
{\small\bfseries Section~\thecontentslabel. \ }{\small\bfseries}
{\hfill\quad\small\bfseries\thecontentspage}
\titleformat{\subsection}[runin]{\bfseries\filright}{\thesubsection.}{.5em}{}[.\ --- \ ]
\titlespacing{\subsection}{0pt}{0.5em}{0pt}
\titlecontents{subsection}[1pc]{\filright}
{\small\thecontentslabel. \ }{ }
{\titlerule*[1pc]{.}\quad\small\thecontentspage}
\titleformat{name=\subsubsection}[runin]{\itshape\filright}{\thesubsubsection.}{.5em}{}[\ --- \ ]
\definecolor{grey}{rgb}{0.6,0.6,0.6}
\def\1#1{{\bf #1}}
\def\2#1{{\cal #1}}
\def\4#1{{\tt #1}}
\def\5#1{{\sf #1}}
\def\6#1{{\mathfrak #1}}
\def\7#1{{\Bbb #1}}
\def\8#1{{\rm #1}}
\def\9#1{{\mathcurl #1}}

\newtheorem{thm}{Theorem}[section]
\newtheorem{defi}[thm]{Definition}
\newtheorem{lem}[thm]{Lemma}
\newtheorem{prop}[thm]{Proposition}
\newtheorem{cor}[thm]{Corollary}
\newtheorem{remark}[thm]{Remark}
\newtheorem{example}[thm]{Example}
\def\fin{\hfill $\lozenge$}
\def\I{\openone}

\def\weyl#1{{\bf w}(#1)}
\def\Alg{\mathcal{A}}
\def\alg#1{\Alg(#1)}
\def\glph#1{\Xi_{#1}}
\def\Qcaph{\vartheta}
\def\qcaph#1{\Qcaph(#1)}
\def\ffunct#1{C_{\field}(\7Z^{#1})}
\def\fffunct#1{C_{\field}(\7T_N^{#1})}
\def\ugroup#1{\mathcal U (#1)}

\def\fring#1{{\mathcal P}_{#1}}
\def\subfring#1{{\mathcal R}_{#1}}
\def\vari{{u}}
\def\symvari#1{{w_{#1}}}
\def\Vari#1#2{\vari_{#1},\cdots,\vari_{#2},\vari^{-1}_{#1},\cdots,\vari^{-1}_{#2}}
\def\poly#1#2{#1[#2]}

\def\ccr#1#2{\6A}

\def\field{{\Bbb F}}
\def\Field#1{{\Bbb Z}_{#1}}

\def\SL#1#2{{\rm SL}(#1,#2)}
\def\convol{\star}
\def\radj#1{\overline{{#1}}}
\def\gen#1{{\1g({#1})}}
\def\genp#1{\1f_{#1}}
\def\genf{\1f}
\def\shear#1{{\left(\begin{array}{cc}1&0\\{#1}&1\end{array}\right)}}
\def\ftra#1{{\left(\begin{array}{cc}0&-#1\\1/#1&0\end{array}\right)}}

\def\deg#1{{\rm deg}({#1})}
\def\coef#1#2{\langle{#2}\rangle_{#1}}
\def\supp#1{{\rm supp}(#1)}
\def\pal#1#2{#1(\vari^{#2}+\vari^{-{#2}})}
\def\uroot#1{\varepsilon}

\def\blf#1#2{\beta(#1,#2)}
\def\Blf{\beta}
\def\blfb#1#2{\tilde{\beta}(#1,#2)}
\def\Blfb{\tilde{\beta}}
\def\spf#1#2{\sigma(#1,#2)}
\def\Spf{\sigma}
\def\mspf#1#2{\tilde\sigma(#1,#2)}
\def\Mspf{\tilde\sigma}
\def\fmspf#1#2{\hat{\sigma}(#1,#2)}
\def\FMspf{\hat{\sigma}}

\def\tra#1{\tau_{#1}}
\def\Tra#1{\alpha_{#1}}

\def\oned#1#2{{\bf r}_{#1}{#2}}
\def\ket#1{|#1\rangle}
\newpagestyle{main}[\small\rm]{
  \headrule
  \sethead[On the structure of Clifford quantum cellular automata][][\usepage]%
  {\usepage}{}{D.-M.~Schlingemann, H.~Vogts and R.~F.~Werner}
  \setfoot[][][]{}{}{}}
\def\Ir{\7Z}
\def\Cx{\7C}
\def\Nei{\2N}

\begin{document}
%%%
%%%
\pagestyle{main}
%%%

%%%layout-titlepage%%%
\thispagestyle{empty}

\vspace*{\stretch{1}}

\begin{center}
\hrule

\vspace*{1cm}

{\large\bfseries On the structure of Clifford quantum cellular automata}

\vspace*{10mm}

Dirk-M.~Schlingemann$^{\ref{aff-1},\ref{aff-2}}$, Holger Vogts$^{\ref{aff-2}}$ and Reinhard~F.~Werner$^{\ref{aff-2}}$

\vspace*{1cm}

\hrule
\end{center}

\vspace*{\stretch{2}}
\begin{center}
{\footnotesize
\begin{enumerate}
\item
ISI Foundation, Quantum Information Theory Unit\\
Viale S. Severo 65\\
10133 Torino, Italy
\label{aff-1}
\item
Institut f\"ur Mathematische Physik\\
Technische Universit\"at Braunschweig\\
Mendelssohnstra{\ss}e~3\\
38106 Braunschweig, Germany
\label{aff-2}
\end{enumerate}
}
\vspace*{1cm}
{\sc \today}
\end{center}
\vspace*{1cm}
\setcounter{page}{0}
\newpage

%%%
%%%
\begin{abstract}
We study reversible quantum cellular automata with the restriction
that these are also Clifford operations. This means that tensor
products of Pauli operators (or discrete Weyl operators) are mapped
to tensor products of Pauli operators. Therefore Clifford quantum
cellular automata are induced by symplectic cellular automata in
phase space. We characterize these symplectic cellular automata and
find that all possible local rules must be, up to some global shift,
reflection invariant with respect to the origin. In the one
dimensional case we also find that every uniquely determined and
translationally invariant stabilizer state can be prepared from a
product state by a single Clifford cellular automaton timestep,
thereby characterizing these class of stabilizer states, and we show
that all 1D Clifford quantum cellular automata are generated by a
few elementary operations. We also show that the correspondence
between translationally invariant stabilizer states and
translationally invariant Clifford operations holds for periodic
boundary conditions.
\end{abstract}
%%%

%%%
\tableofcontents
%%%

%%%
%%%
\section{Introduction}
%%%
%%%

Classical cellular automata have become a standard modeling tool for
complex phenomena. With their discrete time step and their
intrinsically high degree of parallelization they are ideally suited
for models of diverse phenomena as coffee percolation, highway
traffic and oil extraction from porous media. As an abstract
computational model cellular automata can simulate Turing machines,
and even explicit simple automata such as Conway's life game have
been shown to support universal computation \cite{Conw76}. On the
quantum side, the interest in cellular automata stems from their
implementation in optical lattices and arrays of optical microtraps.
However, the theory of quantum cellular automata (QCAs) is still in
its early stage. Since each cell may influence several others, the
dynamics is subject to a ``no-cloning'' constraint, leading to a
non-trivial interplay between the conditions of locality and
unitarity.

It is therefore helpful to have some class of QCAs, which can be
analyzed in great detail, and which can serve as a testing ground
for general ideas about QCAs. This paper is concerned with such an
analysis, namely of the special class of Clifford quantum cellular
automata (CQCAs), in which the elementary time step is given by a
``Clifford gate'', meaning that it takes tensor products of Pauli
matrices to tensor products of other Pauli matrices. In the theory
of gate model computation, and for the one-way quantum computation
model, a detailed analysis of what can be done with Clifford
operations alone turned out to be very useful, even though -- as the
downside of allowing an efficient classical description -- such
gates alone do not allow universal quantum computation. By analogy
it is therefore clear that CQCAs do not comprise the full complexity
of QCAs. What one can hope to get, however, is an interesting class
of cellular automata, and some tools for understanding this class in
great detail.

A similar analysis has been done with Gaussian quantum cellular
automata~\cite{KruWer07}, e.g. the QCA describes a chain of harmonic
oscillators with nearest neighbor couplings. For all these QCAs the
Hilbert space of one elementary cell is infinite dimensional, and
the QCA maps phase space translations, also referred to as Weyl
operators, to phase space translations. In our approach we use
elementary cells with a finite number of levels, which corresponds
to replacing the continuous phase space by a discrete space.

%%%
\subsection{Definition of Clifford quantum cellular automata}\label{def-cqca}
%%%
By definition, a cellular automaton is a lattice system, which
consists of many subsystems (called ``cells'') labeled by a point
lattice in space. For simplicity, we will always take the lattice as
$\7Z^s$, the integer cubic lattice in $s$ space
dimensions\footnote{see however Section~\ref{sec-periodic}, where we
discuss periodic boundary conditions, and hence toroidal lattices}.
The cell systems in the classical case may have states like
``occupied'' and ``empty''. In the quantum case, they will be
$p$-state quantum systems, for some finite $p$. In either case the
group of lattice translations (``shifts'') is a symmetry of the
system.

The dynamics will be given by a discrete global time step, or global ``transition rule'' assumed to have the following three properties:
\begin{itemize}
\item
{\it translation invariance}: the time step commutes with the
lattice translation symmetries.

\item
{\it reversibility}: there is an inverse rule. For a finite quantum
system this would mean unitary dynamics. For the infinite lattice
system this will be stated algebraically below.

\item
{\it locality}, or ``finite propagation speed'': the state of each
cell after one step can be computed from the state in a fixed finite
region around the cell.
\end{itemize}

\noindent These assumption define the class of reversible QCAs
\cite{SchuWer04}. The locality and reversibility conditions are best
phrased in the Heisenberg picture: if $A$ denotes some observable of
the system, its expectation after one time step starting from the
initial state $\rho$ will be $\langle T(A)\rangle_\rho$ for a
suitable observable $T(A)$, where by $\langle A\rangle_\rho$ we
denote the expectation of $A$ in the state $\rho$. The
transformation $A\mapsto T(A)$ is what we will call the {\it global
rule} of the automaton. Then reversibility (together with complete
positivity, which is required of any dynamical map) implies that $T$
is a homomorphism of the observable algebra of the whole system: $T$
is linear, $T(AB)=T(A)T(B)$, and $T(A^\ast)=T(A)^\ast$. Locality
means that an observable $A_x$ localized at some lattice point
$x\in\Ir^s$ (i.e., an observable of the cell at $x$) will be mapped
to an observable localized in the region $x+\Nei$. That is, $T(A_x)$
will be in the tensor product of the cell algebras belonging to the
sites $x+n$ with $n\in\Nei$. By translation invariance this set
$\Nei$, called the {\it neighborhood scheme} of the automaton is
independent of $x$.

The global transition rule $T$ is a map on an infinite dimensional
space, and hence not readily specified explicitly. However, by using
the basic properties of QCAs one can see that it suffices to know
just a few local data, associated with the region $\Nei$, in order
to reconstruct $T$ uniquely. Suppose we know $T(A_x)$ for every
observable $A_x$ in some basic cell $x$. Then by translation
invariance we know the analogous transformation for {\it any} cell.
Moreover, since every local observable can expanded in products of
one-cell observables, and because $T$ is a homomorphism, we can
compute $T$ for any local observable. So the restriction $T_x$ of
$T$ to the observables of a single cell $x$ can be called the {\it
local rule} of the automaton. We can also decide by a finite set of
equations, whether a proposed local rule actually belongs to a
well-defined global rule: clearly $T_x$ must be a homomorphism. The
only further condition one has to check is that the images of
observables $A_x$ and $B_y$ localized in the cells indicated also
commute, i.e., $T_x(A_x)T_y(B_y)=T_y(B_y)T_x(A_x)$, whenever $x\neq
y$. This is necessary, because $A_x$ and $B_y$ commute, and a
moment's reflection shows that this is also sufficient for uniquely
reconstructing the images of arbitrary observables under $T$. The
commutation conditions on the local rule are trivially satisfied
when $x$ and $y$ are sufficiently far apart, when $x+\Nei$ and
$y+\Nei$ are disjoint. Hence only finitely many conditions need to
be checked.

For special classes, the job of specifying a QCA via its local rule
can be reduced still further, which is where the Clifford condition
comes in. Let us assume now that we have a qubit system, so the
local cell dimension is $p=2$. For each local cell we thus have a
basis for the observables, consisting of the identity and the three
Pauli matrices, which we denote by $X,Y$ and $Z$. By $X_x$ etc.\ we
denote the corresponding Pauli matrix belonging to the cell $x$.
Finite tensor products of Pauli matrices belonging to different
sites, perhaps with a sign $\pm1,\pm i$ will be referred to as {\it
Pauli products}. These form a group, called the Pauli group. Then a
{\it Clifford quantum cellular automaton} (CQCA for short) is
defined by the condition
\begin{itemize}
\item {\it Clifford condition}: If $A$ is any multiple of a Pauli product, so is $T(A)$.
\end{itemize}
Clearly, this is equivalent to the property that one-cell Pauli
operators are taken to Pauli products, which simplifies the local
rule. Moreover, it suffices to specify $T(X_x)$ and $T(Z_x)$ for
some $x$, because we can compute $T(Y_x)=T(iX_xZ_x)=iT(X_x)T(Z_x)$
via the homomorphism property. Hence a CQCA is defined in terms of
just two Pauli products.

%%%
\begin{example}
\label{example-0} \rm For the one-dimensional lattice ($s=1$),
consider the relations
\begin{equation}\label{example-rule}
 \begin{array}{lclcr}
  T(X_x)&=& &-Z_x&\\
  T(Z_x)&=&Z_{x-1}\otimes &X_x&\otimes Z_{x+1}
 \end{array}
\end{equation}
Let us verify that all requirements for a local rule are satisfied.
To begin with each of the expressions on the right hand side, as a
product of Pauli matrices, is hermitian with square one. These are
all the required conditions related to just a single line, and are
satisfied for any Pauli product with a sign $\pm1$. Next we have to
verify the anti-commutation relation arising from applying a
homomorphism $T$ to the anti-commutation relation $XZ+ZX=0$. Indeed,
$T(X_x)T(Z_x)+T(Z_x)T(X_x)=-Z_{x-1}\otimes(Z_xX_x+X_xZ_x)\otimes
Z_{x+1}=0$. Hence the definition $T(Y_x)=iT(X_x)T(Z_x)$ again
produces a hermitian operator with square $\I$, and the local rule
is a homomorphism $T_x$ into the algebra on the sites $x+\Nei$ with
$\Nei=\{-1,0,1\}$. Finally, we have to check the commutation rules
for the images of observables from neighboring sites. For example,
we have $[T(X_x),T(Z_{x+1})]=-[Z_x,Z_x\otimes X_{x+1}\otimes
Z_{x+2}]=0$, and similarly $[T(Z_x),T(Z_{x+2})]=0$. Perhaps the only
non-trivial relation to check is
\begin{displaymath}
  [T(Z_x),T(Z_{x+1})]
  =[Z_{x-1}\otimes X_x\otimes Z_{x+1},Z_{x}\otimes X_{x+1}\otimes Z_{x+2}]
  =0,
\end{displaymath}
which holds because the factors on sites $x$ and $x+1$ both {\it anti-}commute.

In principle, we would also have to check the existence of an
inverse for the automaton, which is actually given by
$T(X_x)=X_{x-1}\otimes Z_x\otimes X_{x+1}$ and $T(Z_x)=-X_x$, but as
was shown in \cite{SchuWer04}, this already follows from the
homomorphism property. \fin
\end{example}

It is clear from this example that the search for CQCAs is now a
combinatorial problem. We can first look for {\it self commuting}
Pauli products, i.e., Pauli products, which commute with all
translates of itself. Only these can appear on the right hand side
of local rules. One can then check, for any pair $X',Z'$ of such
products, whether they anti-commute, while all proper translates of
$X'$ commute with $Z'$. In fact, we began our investigation by
running this simple search program. We found, for example, that
while there is a rich variety of self-commuting Pauli products only
reflection symmetric products could appear in a local rule. This
will indeed be shown in full generality below.

%%%
\subsection{Translationally invariant stabilizer states}
%%%
Commuting sets of Pauli products also play a central role in the
problem of determining so-called stabilizer states: these are pure
states, which can be characterized by eigenvalue equations for Pauli
products or, equivalently, by the condition that certain Pauli
products have expectation $\pm1$. It is easy to check that Pauli
products which simultaneously have sharp expectations $\pm1$ must
commute. Now for the infinite lattice systems it is natural to ask
which Pauli products $A$ have the property that there is a unique
pure state $\rho$ of the infinite system, which has expectation 1
for $A$ and all its translates.

As the simplest example, let us take $A=Z_x$, so we ask for states
with $\langle Z_x\rangle=1$ for all $x\in\Ir^s$. Clearly, this
defines the ``all spins up'' state, which is an infinite product
state. A slightly more complex example uses the stabilizer operators
$A=Z_{x-1}\otimes X_x\otimes Z_{x+1}$, which singles out the
one-dimensional {\it cluster state}, whose higher dimensional
analogs are used as the entanglement resource for universal one-way
quantum computing~\cite{BrieRau01c}.

Showing that these eigenvalue equations define a unique state of the
infinite lattice is now very easy, by using the cellular automaton
(\ref{example-rule}): Since this automaton maps $Z_x$ to the
required stabilizer operator, all existence and uniqueness problems
for such a state are mapped to the corresponding trivial questions
for the stabilizer operator $Z_x$. In other words, self-commuting
Pauli products of the form $A=T(Z_x)$ for some CQCA $T$ characterize
a unique translation invariant cluster state. We will show later
that (at least in one dimension) the converse is also true, so that
there is a very close connection between stabilizer states and
Clifford cellular automata.

%%%
\subsection{Our methods and techniques}
%%%
The definition of CQCAs given above applies only to qubit systems.
However, all our results are also valid for higher dimensional
cells, particular cells of prime dimension $p$. The role of the
Pauli operators $X$ and $Z$ is then taken by the cyclic shift on
$\Cx^p$, and the multiplication by a phase, i.e.
\begin{equation}\label{onestep}
 \begin{array}{lcr}
  X\ket{q}&=& \ket{q+1}\\
  Z\ket{q}&=& \8e^{2\pi i q/p}\ket{q},
 \end{array}
\end{equation}
where all ket labels $q$ are taken modulo $p$. Products of these
operators are called {\it Weyl operators}, and the appropriate
definition of CQCAs requires that $T(X_x)$ and $T(Z_x)$ are both
tensor products of Weyl operators. The necessary preliminaries on
the Pauli group and Clifford operations in this extended setting,
and the background concerning infinite lattice systems are provided
in Subsection~\ref{sec-2-1}.

In order to utilize the translation symmetry one would like to use
Fourier transform techniques. However, in the discrete structures an
integral with complex phases makes no sense. It turns out, however,
that a ``generating function'' technique does nearly as well. The
analogue of the Fourier transform is then a Laurent-polynomial in an
indeterminate variable, i.e., a polynomial with coefficients in the
field $\Field{p}=\7Z_p$ with both positive and negative powers. The
salient facts about this structure will be provided in
Subsection~\ref{sec-2-3}.

The description in terms of Laurent polynomials can also be adapted
to lattices with periodic boundary conditions. This will be
described in Section~\ref{sec-periodic}.

%%%
\subsection{Outline and summary of results}
%%%
In order to discuss general Clifford quantum cellular automata, that
is, for arbitrary lattice and single cell dimension, we introduce in
Section~\ref{sec-2} the necessary mathematical tools. We first
review the concept of discrete Weyl systems and (infinite) tensor
products of them, thereby characterizing the underlying ``phase
space''. We show that Clifford QCAs can be completely characterized
in terms of classical symplectic cellular automata. We also
introduce our Fourier transform techniques and study the structure
of isotropic subspaces, because these play an essential role for the
characterization of symplectic cellular automata and translationally
invariant stabilizer states.

In Section~\ref{sec-3} we will state our main results. We show that
symplectic cellular automata can be identified with two-by-two
matrices, which have Laurent-polynomials as matrix elements. We will
find that the determinant of this matrix must be one and that the
polynomials must be reflection invariant. In the one-dimensional
case we state that every translationally invariant stabilizer state
can be prepared out of a product state by a single CQCA step.
Furthermore, we also specify the generators of all 1D QCAs.

Finally, we show in Section~\ref{sec-periodic} that the close
connection between translationally invariant stabilizer states and
CQCAs also holds in the case of periodic boundary conditions even in
every lattice dimension.

%%%
%%%
\section{Mathematical tools}
%%%
%%%
\label{sec-2}
We introduce some mathematical tools, which we will
use to study Clifford QCAs. We start with a short repetition of
finite Weyl systems, which generalize the Pauli operators to systems
with prime number dimensions. These Weyl operators can be described
by phase space vectors and Clifford operations are induced by
symplectic transformations on the phase space. Since we are looking
for translationally invariant operations, we also introduce some
kind of Fourier transform.
%%%
\subsection{Weyl algebras}
%%%
\label{sec-2-1}

Each single cell in a QCA is given by a finite dimensional quantum
system, so the observables on a single system can be described by
matrices from the algebra $\2M_p(\7C)$. A possible basis for this
algebra is given by Weyl operators $\weyl{r,k}=X^rZ^k$, whereby $X$
and $Z$ are given by the generalized Pauli operators from
equation~(\ref{onestep}). These operators fulfill the Weyl relations
\begin{equation}
\weyl{r_1+r_2,k_1+k_2}=\uroot{p}^{-r_2k_1}\weyl{r_1,k_1}\weyl{r_2,k_2}\,,
\end{equation}
where $\uroot{p}=\exp(2\pi\8i/p)$ is the $p^{\rm th}$ root of unity.
>From this equation the commutation relation
\begin{equation}
\weyl{r_1,k_1}\weyl{r_2,k_2}=\uroot{p}^{r_1k_2-r_2k_1}\weyl{r_2,k_2}\weyl{r_1,k_1}
\end{equation}
immediately follows. Obviously we get for $p=2$ the standard Pauli
operators from
\begin{equation}
X=\weyl{1,0}\,,\quad Y=i\weyl{1,1}\,,\quad Z=\weyl{0,1}\;,
\end{equation}
and the Weyl operators are generalizations of the Pauli operators to
higher dimensional spaces. The indices $r$ and $k$ are integers
modulo $p$, so they are elements of the finite field
$\field=\Field{p}$. In infinite dimensional systems Weyl operators
describe phase space translations and therefore we call the space
$\field^2$ a discrete phase space.

Building a tensor product of Weyl operators means that we must
assign a phase space vector $\xi(x)=(\xi_+(x),\xi_-(x))\in\field^2$
to each lattice point $x\in\7Z^s$, so $\xi$ is a mapping from
$\7Z^s$ into $\field^2$ and we denote for the tensor product
\begin{equation}
\weyl{\xi}=\bigotimes_{x\in\7Z^s}\weyl{\xi(x)}\;.
\end{equation}
This infinite tensor product is well defined, if there are only
finitely many of the Weyl operators different from $\I=\weyl{0}$.
For a mapping $\xi:\7Z^s\to\field^2$ we have that only finitely many
$x$ with $\xi(x)\ne 0$ are allowed, so the support of $\xi$ is
finite. The set of such functions describes the global system and is
identified with the global phasespace $\glph{s}$. We denote the
finitely supported functions from $\7Z^s$ to $\field$ by
$\ffunct{s}$ and we have $\glph{s}=\ffunct{s}^2$. The corresponding
Weyl operators generate an algebra and, by restricting the support
of the functions to some finite subset $\Lambda\in\7Z^s$, we get a
finite dimensional algebra
$\alg\Lambda=\bigotimes_{x\in\Lambda}\2M_p(\7C)$, also called the
local algebra of $\Lambda$. By taking the union of these algebras
over all finite subsets of $\7Z^s$ and taking the closure (in
operator norm) we get a quasilocal $C^\ast$-algebra
$\Alg$~\cite{BraRob}, which is used in the general theory of quantum
cellular automata~\cite{SchuWer04}.

The local structure is accompanied by the symmetry group of lattice
translations. For each lattice translation $x\in\7Z^s$ an
automorphism $\Tra{x}$ is defined by
\begin{equation}
\label{translation} \Tra{x}\weyl{\xi}=\weyl{\tra{x}\xi} \; .
\end{equation}
where $\tra{x}$ is the translation of phase space vectors. Given a
phase space vector $\xi$, the translated vector is
$(\tra{x}\xi)(y)=\xi(y-x)$. So the automorphism $\Tra{x}$ shifts the
position of each tensor factor by $x$. It follows directly from
(\ref{translation}) that the homomorphism property
$\Tra{x+y}=\Tra{x}\Tra{y}$ holds. Furthermore, the automorphism
$\Tra{x}$ maps the local algebra $\alg\Lambda$ onto
$\alg{\Lambda+x}$.

The Weyl relations of a single system completely determine the
relations of the global system which are given by
\begin{equation} \label{ccr}
\weyl{\xi+\eta}=\uroot{p}^{\blf{\xi}{\eta}}\weyl{\xi}\weyl{\eta}\;,
\end{equation}
where we have introduced the bilinear form
$\blf{\xi}{\eta}:=\sum_{x\in\7Z^s}\xi_+(x)\eta_-(x)$. The adjoint of
a Weyl operator is given by
\begin{equation}
\label{star} \weyl{\xi}^\ast=\uroot{p}^{-\blf{\xi}{\xi}}\weyl{-\xi}
\end{equation}
which is due to the unitarity of the Weyl operators.

Since commutation relations are essential for validating possible
local rules of quantum cellular automata, the commutation relations
of Weyl operators are most important for us. We get
\begin{equation}
\label{com-ccr}
\weyl{\eta}\weyl{\xi}=\uroot{p}^{\spf{\xi}{\eta}}\weyl{\xi}\weyl{\eta}
\; ,
\end{equation}
whereby $\spf{\xi}{\eta}:=\blf{\xi}{\eta}-\blf{\eta}{\xi}$ is the
canonical symplectic form on $\glph{s}$. This means that two Weyl
operators $\weyl{\xi}$ and $\weyl{\eta}$ are commuting if and only
if $\spf{\xi}{\eta}=0$ (and for $p=2$ they anti-commute if
$\spf{\xi}{\eta}=1$). In particular, an abelian algebra of Weyl
operators is given by a subspace of $\glph{s}$ on which the
symplectic form vanishes. Such a subspace is called isotropic and a
maximally abelian algebra corresponds to a maximally isotropic
subspace.

%%%
\subsection{Clifford quantum cellular automata}
%%%
\label{sec-2-2} As already mentioned a Clifford quantum cellular
automaton is a QCA which maps Weyl operators to multiples of Weyl
operators, which are in our case tensor products of single cell Weyl
operators, so we have the relation (the ``Clifford condition'')
\begin{equation}\label{cqca}
T(\weyl{\xi})=\qcaph{\xi}\weyl{\1t\xi}\;
\end{equation}
with a mapping $\1t$ on the phase space $\glph{s}$ and some phase
valued function $\Qcaph:\glph{s}\to\ugroup1=\{z\in\7C||z|=1\}$.
Since $T$ is an automorphism we find with equation~(\ref{com-ccr})
that
$\weyl{\1t\xi}\weyl{\1t\eta}=\uroot{p}^{\spf{\eta}{\xi}}\weyl{\1t\eta}\weyl{\1t\xi}$
holds, so we have $\spf{\1t\xi}{\1t\eta}=\spf{\xi}{\eta}$ or in
other words $\1t$ is a symplectic transformation.

For reversible operations the Clifford condition is in general
equivalent to the Weyl covariance (for general theory on covariant
channels we refer to \cite{Scu79} and for the special case of Weyl
covariance to \cite{Hol02,Hol04}) of the quantum channel:

%%%
\begin{prop}
An automorphism $T$ on the Weyl algebra $\Alg$ fulfills the Clifford
condition~(\ref{cqca}) if and only if the Weyl covariance
\begin{equation}
T(\weyl{\eta}A\weyl{\eta})=\weyl{\1t\eta}T(A)\weyl{\1t\eta}^\ast\quad\forall\eta\in\glph{s}
\end{equation}
holds for all operators $A\in\Alg$ and some symplectic
transformation $\1t$.
\end{prop}
%%%
\begin{proof}
Because the Weyl operators form a basis of $\Alg$ we just have to
insert $\weyl{\xi}$ for some $\xi\in\glph{s}$ in the covariance
condition, which yields the equation
$\uroot{p}^{\spf{\xi}{\eta}}T(\weyl{\xi})=\weyl{\1t\eta}T(\weyl{\xi})\weyl{\1t\eta}^\ast$.
If $T$ is a Clifford automorphism we have already seen that $\1t$ is
a symplectic transformation and obviously fulfills this equation. In
the inverse direction we get that $T(\weyl{\xi})$ must be a multiple
of $\weyl{\1t\xi}$, because the relation must hold for all
$\eta\in\glph{s}$ and the symplectic form is non degenerate (note
that the support of the phase space vectors is finite and that $\1t$
maps therefore finitely supported vectors to finitely supported
vectors, so the commutation relations can be checked in a finite
dimensional space).
\end{proof}
%%%

Since a QCA is a translationally invariant automorphism on the
quasilocal algebra, it suffices that the Clifford condition holds
for the local rule, e.g. the QCA restricted to operators which are
localized in a single cell. Furthermore, because of the Weyl
relations on a single cell, we only need to specify the image of the
Weyl operators $\weyl{1,0}$ and $\weyl{0,1}$. To some extend we are
free in the choice of the phases
$\qcaph{1,0},\qcaph{0,1}\in\ugroup1$, since these phases do not
interfere with the commutation relations for the local rule. The
only condition is that some power of a Weyl operator is always equal
to $\I$ (we will specify this below), and so these phases must be
some roots of unity. The two phases $\qcaph{1,0}$ and $\qcaph{0,1}$
completely determine the function $\Qcaph$.

Of course $\1t$ and $\Qcaph$ must be translationally invariant,
because $T$ is translationally invariant. Using the homomorphism
property of the QCA and equation~(\ref{ccr}) we get
$\qcaph{\xi+\eta}\weyl{\1t(\xi+\eta)}=\qcaph{\xi}\qcaph{\eta}\uroot{p}^{\blf{\xi}{\eta}-\blf{\1t\xi}{\1t\eta}}\weyl{\1t\xi+\1t\eta}$,
so -- because the Weyl operators form a basis -- the transformation
$\1t$ must be linear and the phase function must fulfill
\begin{equation}\label{coboundary}
\qcaph{\xi+\eta}=\qcaph{\xi}\qcaph{\eta}\uroot{p}^{\blf{\xi}{\eta}-\blf{\1t\xi}{\1t\eta}}\;,
\end{equation}
which enables us to calculate the phase $\qcaph{\xi}$ for each
$\xi\in\glph{s}$, if the local rule and therefore $\1t$ and the
phases $\qcaph{1,0}$ and $\qcaph{0,1}$ are given. In total we get
the following theorem:

%%%
\begin{thm} \label{thm-1} If $T$ is a Clifford quantum cellular automaton (equation (\ref{cqca})) on the Weyl algebra
$\Alg$, then $\1t$ is a translationally invariant linear symplectic
transformation (``symplectic cellular automaton'') and the phase
function $\Qcaph$ fulfills equation~(\ref{coboundary}).
\end{thm}
%%%

This means that we are able to study Clifford QCAs -- up to some
phase function -- in terms of a classical cellular automaton on the
phase space $\glph{s}$. It is well known that Clifford operations
allow an efficient classical description, which in the case of QCAs
turned out to be the group of classical symplectic cellular
automata. In the rest of the paper we will study the structure of
this kind of cellular automata, thereby characterizing the structure
of CQCAs.

We would like to give a closed expression for the phase function
$\Qcaph$, but this has to be done in dependence of the cell
dimension. First we consider the case $p\ne 2$. Then all Weyl
operators fulfill $\weyl{\xi}^p=\I$ and because of $T(\I)=\I$ the
phase $\qcaph{\xi}$ must be a $p^{\rm th}$ root of unity. So we can
write $\qcaph{\xi}=\uroot{p}^{\varphi(\xi)}$ with a function
$\varphi:\glph{s}\to\field$, which then has to fulfill
$\varphi(\xi+\eta)=\varphi(\xi)+\varphi(\eta)+\blf{\xi}{\eta}-\blf{\1t\xi}{\1t\eta}$.
This equation determines the function $\varphi(\xi)$ up to some
linear functional $\lambda(\xi)$, which is given by the choice of
the phases $\varphi(1,0)$ and $\varphi(0,1)$. The bilinear form
$\blf{\xi}{\eta}-\blf{\1t\xi}{\1t\eta}$ is symmetric, because $\1t$
is a symplectic transformation. If $p\ne 2$ we may divide by $2$ and
the general solution is
$\varphi(\xi)=\frac{1}{2}(\blf{\xi}{\xi}-\blf{\1t\xi}{\1t\xi})+\lambda(\xi)$.

The case of qubits ($p=2$) is slightly more complicated because the
Weyl operators fulfill $\weyl{\xi}^2=(-1)^{\blf{\xi}{\xi}}\I$. So
the phase function must fulfill $\qcaph{\xi}=\8i^{\varphi(\xi)}$
with $\varphi:\glph{s}\to\Field{4}$. We replace the form
$\Blf:\glph{s}\times\glph{s}\to\Field{2}$ by the bilinear form
$\Blfb:\glph{s}\times\glph{s}\to\Field{4}$, which is formally given
by $\Blfb=2\Blf$, so the values of $\Blfb$ are even elements of
$\Field{4}$ and the Weyl relation becomes
$\weyl{\xi+\eta}=\8i^{\blfb{\xi}{\eta}}\weyl{\xi}\weyl{\eta}$. This
means that $\varphi$ fulfills
$\varphi(\xi+\eta)=\varphi(\xi)+\varphi(\eta)+\gamma(\xi,\eta)$ with
the form $\gamma(\xi,\eta)=\blfb{\xi}{\eta}-\blfb{\1t\xi}{\1t\eta}$.
This form is symmetric, so in the decomposition
$\gamma(\xi,\eta)=\sum_{i,j}\gamma_{ij}\xi_i\eta_j$ we have
$\gamma_{ij}=\gamma_{ji}$ and all these elements are even. We can
find $\gamma_i$ with $\gamma_{ii}=2\gamma_i$, but this choice is not
unique in $\Field{4}$ and corresponds exactly to the freedom in the
choice of the phases $\qcaph{1,0}$ and $\qcaph{0,1}$. The solution
for $\varphi$ is then given by
$\varphi(\xi)=\sum_{i<j}\gamma_{ij}\xi_i\xi_j+\sum_i\gamma_i\xi_i$
(note that $\xi_i\in\{0,1\}$ and so $\xi_i^2=\xi_i$ holds).

%%%
\subsection{Algebraic Fourier transform}
%%%
\label{sec-2-3}

We would like to use Fourier transform techniques for the study of
the structural properties of symplectic CA, because of translational
invariance, and because we know that this is very helpful for
symplectic CA with continuous single cell phase
space~\cite{KruWer07}. So we have to apply a Fourier transform to
the functions $\ffunct{s}$. But the values of these functions are in
the finite field $\field$ and multiplying such a value with a
complex number does not really match. It turns out that a slight
modification of the usual Fourier transform does as well. For a
function $f\in\ffunct{s}$ we define
\begin{equation}\label{fourier}
\hat f(\vari)=\sum_{x\in\7Z^s}f(x)\vari^x\;,
\end{equation}
with $\vari^x=\vari_1^{x_1}\cdots \vari_s^{x_s}$. Now the
transformed function $\hat f$ is a polynomial or, more precisely, a
Laurent-polynomial in the variables $\vari_1,\dots,\vari_s$ with
coefficients in $\field$, which will be denoted by
$\fring{s}=\poly{\field}{\Vari{1}{s}}$. Note that we have indeed
polynomials, because the functions in $\ffunct{s}$ are finitely
supported. Equation~(\ref{fourier}) identifies functions of
$\ffunct{s}$ with polynomials $\fring{s}$ and this identification is
unique, so $\ffunct{s}$ and $\fring{s}$ are isomorphic. The usual
Fourier transform would require $\vari_n=e^{ip_n}$. We do not
further specify the domain of the variables, and this approach can
be seen as ``generating function approach'' or ``algebraic Fourier
transform''.

The convolution $f\convol h=\sum_x f(-x)\tra{x}h$ is a natural
product\footnote{With the convolution the set $\ffunct{s}$ becomes a
``commutative division ring''.} of functions in $\ffunct{s}$. The
invertible elements with respect to this operation are the functions
which are supported on a single lattice point, e.g. $f=c\delta_x$
($\delta_x$ is the Kronecker-delta) with $c\in\field$ and
$x\in\7Z^s$, and the unit element is $\delta_0$. The nice fact about
Fourier transform is that the convolution turns into a usual product
which is also true for our algebraic version:
\begin{equation}
\widehat{f\convol h}=\hat{f}\hat{h}\quad f,h\in\ffunct{s}\;.
\end{equation}
Note that the invertible polynomials are monomials\footnote{This
will be different when we go to periodic boundary conditions.}, e.g.
they are of the form $u^x$. Of course the unit element is the
constant $1=\hat\delta_0$. Another important operation is the
reflection operation (or involution) $\radj{f}(x):=f(-x)$ for
$f\in\ffunct{s}$. Obviously the reflection preserves the
convolution, e.g. $\radj{f\convol h}=\radj{f}\convol\radj{h}$, and
for the transformed function we have
$\radj{\hat{f}}(u)=\hat{f}(u^{-1})$.

The phase space $\glph{s}$ consists of two-dimensional tuples of
functions from $\ffunct{s}$ and all operations can be defined
component-wise\footnote{With the component-wise convolution the
phase space is a two-dimensional $\ffunct{s}$-module.}, so we get
that the phase space is isomorphic to $\glph{s}\cong\fring{s}^2$. We
would like to study the structure of symplectic CA in this
polynomial space. The transformation of an operation
$\1t:\glph{s}\to\glph{s}$ is defined according to
$\hat{\1t}\hat\xi=\widehat{\1t\xi}$, so $\hat{\1t}$ is a mapping
from $\fring{s}^2$ to $\fring{s}^2$. We introduce the symplectic
form $\FMspf:\fring{s}^2\times\fring{s}^2\to\fring{s}$ by
\begin{equation}
\fmspf{\xi}{\eta}=\radj{\xi_+}\eta_--\radj{\xi_-}\eta_+\;,\quad\xi,\eta\in\fring{s}^2\;,
\end{equation}
which can be written as $\fmspf{\xi}{\eta}=\det(\radj\xi,\eta)$,
whereby $(\xi,\eta)$ denotes the $2\times 2$-matrix
\begin{equation}
(\xi,\eta)=\left(\begin{array}{cc}\xi_+&\eta_+\\\xi_-&\eta_-\end{array}\right)
\end{equation}
with polynomial entries. The symplectic form $\FMspf$ is the best
fitting symplectic form for symplectic CA, because it combines both
the basic symplectic form $\Spf$ as well as the translation
invariance:

%%%
\begin{prop}\label{sca-fourier}
A linear operation $\1t$ on the phase space $\glph{s}$ is a
symplectic cellular automaton, if and only if, the transformed
operation $\hat{\1t}$ leaves the symplectic form $\FMspf$ invariant.
\end{prop}
%%%
\begin{proof}
For this proof we introduce the form
$\mspf{\xi}{\eta}=\spf{\xi}{\tra{(\cdot)}\eta}$ for
$\xi,\eta\in\glph{s}$. A straightforward computation shows that
$\mspf{\xi}{\eta}=\radj{\xi_+}\convol \eta_--\radj{\xi_-}\convol
\eta_+$ holds. This means we have
$\widehat{\mspf{\xi}{\eta}}=\fmspf{\hat\xi}{\hat\eta}$, so $\FMspf$
is the Fourier transform of $\Mspf$ and the invariance of $\Mspf$
under some operation $\1t$ is equivalent to the invariance of
$\FMspf$ under $\hat{\1t}$.

Now suppose $\1t$ is a symplectic CA. Then we have for all
$x\in\7Z^s$ that
$\mspf{\1t\xi}{\1t\eta}(x)=\spf{\1t\xi}{\tra{x}\1t\eta}=\spf{\1t\xi}{\1t\tra{x}\eta}=\spf{\xi}{\tra{x}\eta}=\mspf{\xi}{\eta}(x)$
holds, because $\1t$ is translationally invariant and preserves
$\Spf$, so $\Mspf$ is invariant under $\1t$.

If $\1t$ leaves $\Mspf$ invariant, this holds also for
$\Spf=\mspf{\cdot}{\cdot}(0)$, and because of this
$\spf{\1t\xi}{\tra{x}\1t\eta}=\spf{\xi}{\tra{x}\eta}=\spf{\1t\xi}{\1t\tra{x}\eta}$
holds for all $x\in\7Z^s$ and all $\xi,\eta\in\glph{s}$ and so $\1t$
must commute with the translations $\tra{x}$.
\end{proof}
%%%

So we can characterize symplectic CA in ``momentum space'' by
studying the linear transformations on $\fring{s}^2$ which leave the
symplectic form $\FMspf$ invariant. In the subsequent we will mainly
work in the polynomial space $\fring{s}$. Therefore we will just
identify the phase space $\glph{s}$ with $\fring{s}^2$ and we will
omit the symbol $\hat{\hspace{0.2cm}}$ for the Fourier transform of
transformations.

%%%
\subsection{Isotropic subspaces}
%%%

As we have already seen in the introduction, commutation relations
are important for the verification of local rules of reversible
QCAs, because a QCA is a homomorphism and preserves the algebraic
structure. Especially the images of $X_x$ and $Z_x$ must be
``self-commuting'', meaning that $[T(X_x),T(X_y)]=0=[T(Z_x),T(Z_y)]$
holds for all $x,y\in\7Z^s$. So the operators $T(X_x)$ generate a
translationally invariant abelian algebra. For Weyl operators
translationally invariant abelian algebras correspond exactly to
isotropic subspaces of $\fring{s}^2$ with respect to the symplectic
form $\FMspf$ and these subspaces can be easy connected to
translationally invariant stabilizer states. Therefore it is
important for us to study the structure of these subspaces.

A $\fring{s}$-subspace\footnote{More precisely one should say
submodule, but we will use the more convenient word subspace.}
$\2I\subset\fring{s}^2$ is called isotropic, if for all
$\xi,\eta\in\2I$ the symplectic form $\fmspf{\xi}{\eta}=0$ vanishes.
An isotropic $\fring{s}$-subspace $\2I$ is called maximally
isotropic, if the relation $\fmspf{\xi}{\eta}=0$ for all $\xi\in\2I$
implies that $\eta\in\2I$ holds.

For us the form of the generators of isotropic, in particular
maximally isotropic, $\fring{s}$-subspaces is important, because
this is a substantial step for the characterization of local rules
of CQCAs and translationally invariant stabilizer states. The
following lemma shows that a generator $\xi$ of a singly generated
maximally isotropic subspace is reflection invariant and that the
components $\xi_+$ and $\xi_-$ are coprime. We will call a
polynomial $f\in\fring{s}$ (or a tuple of those) reflection
invariant for some half integer lattice point $a\in
\frac{1}{2}\7Z^s$, if $f=\vari^{2a}\radj f$ holds. The greatest
common divisor of two polynomials $f,h\in\fring{s}$ will be denoted
by $\8{gcd}(f,h)$. Note that the greatest common divisor is defined
only up to invertible elements. We will simply write
$\8{gcd}(f,h)=1$, if $f$ and $h$ are coprime.

%%%
\begin{lem}\hfill\\[-0.8cm]
\label{lem-iso-1}
\begin{enumerate}
\item
If the subspace $\fring{s}\xi\subset\fring{s}^2$ is maximally
isotropic, we have $\8{gcd}(\xi_+,\xi_-)=1$.
\item
If the subspace $\fring{s}\xi\subset\fring{s}^2$ is maximally
isotropic, $\xi$ is reflection invariant to some point $a\in
\frac{1}{2}\7Z^s$.
\item
Every reflection invariant polynomial generates an isotropic
$\fring{s}$-subspace.
\end{enumerate}

\end{lem}
%%%
\begin{proof}
Ad 1. Suppose $\fring{s}\xi$ is a maximally isotropic
$\fring{s}$-subspace and $\8{gcd}(\xi_+,\xi_-)=f$ is not invertible.
So we can write $\xi=f\eta$ with $\8{gcd}(\eta_+,\eta_-)=1$, but
$\eta\notin\fring{s}\xi$ since $f$ is not invertible. But we have
that $\fmspf{\xi}{\eta}=f\fmspf{\xi}{\xi}=0$ holds, which is a
contradiction to $\fring{s}\xi$ being maximally isotropic.

Ad 2. Suppose that $\fring{s}\xi$ is maximally isotropic. By 1 we
have $\8{gcd}(\xi_+,\xi_-)=1$. Since
$\fmspf{\radj{\xi}}{\xi}=\xi_+\xi_--\xi_-\xi_+=0$, it follows that
$\radj{\xi}\in\fring{s}\xi$. So we have $\radj{\xi}=f\xi$ with some
polynomial $f\in\fring{s}$. But for the reflected phase space vector
$\radj{\xi}$ we also have that
$\8{gcd}(\radj{\xi}_+,\radj{\xi}_-)=1$, so $f$ must be invertible
and therefore a monomial $f(\vari)=\vari^{-2a}$ for some
$a\in\frac{1}{2}\7Z^s$.

Ad 3. Suppose that $\xi=\vari^{2a}\radj{\xi}$ is reflection
invariant. Then
$\fmspf{\xi}{\xi}=\fmspf{\vari^{2a}\radj{\xi}}{\xi}=\vari^{-2a}(\xi_+\xi_--\xi_-\xi_+)=0$
holds, and $\xi$ generates an isotropic $\fring{s}$-subspace.
\end{proof}
%%%

%%%
\begin{example}\label{stand-ex}\rm
Both $\xi_1=(1+\vari){0\choose 1}$ and $\xi_2={1\choose
\vari+\vari^{-1}}$ are reflection invariant. The corresponding Weyl
operators $\weyl{\xi_1}=Z_0\otimes Z_1$ and $\weyl{\xi_2}=Z_{-1}
\otimes X_0\otimes Z_1$ are the same reading from the left and from
the right (``palindromes''). Both phase space vectors generate
isotropic subspaces. The subspace generated by $\xi_2$ is indeed
maximally isotropic and the components $\xi_{2,+}$ and $\xi_{2,-}$
are coprime, whereas the subspace generated by $\xi_1$ is not
maximally isotropic because $1+\vari$ is a nontrivial common
divisor. This is also clear in terms of operators, because all
operators $Z_x$ commute with $\weyl{\xi_1}$, but cannot be obtained
by products of translates of $\weyl{\xi_1}$. \fin\end{example}
%%%

In particular, the greatest common divisor comes into play. We will
be able to state more results in the one-dimensional case ($s=1$),
due to the fact that the ring of polynomials
$\fring{}:=\fring{1}=\poly{\field}{\vari,\vari^{-1}}$ is
euclidean\footnote{In more abstract words $\fring{}$ is a principal
ideal ring, which means that every ideal in $\fring{}$ is generated
by a single element. For this general algebraic theory we refer to
\cite{Jac75}.}. Especially this means that the euclidean algorithm
can be applied for finding the greatest common divisor of two
polynomials, which is also used for the factorization of wavelet
transformations~\cite{DauSwe98}.

%%%
\begin{lem}[Extended euclidean algorithm for Laurent
polynomials]\label{lem-euclid} Let $\xi\in\fring{}^2$ be a phase
space vector. Then there exist $f_0, f_1\in\fring{}$ such that
\begin{equation}\label{euclid}
f_0\xi_++f_1\xi_-=\8{gcd}(\xi_+,\xi_-)
\end{equation}
holds.
\end{lem}
%%%
\begin{proof}
We define the degree of a Laurent polynomial
$f=\sum_{x=L_-}^{L_+}f_x\vari^x$ by $\8{deg}(f):=L_+-L_-$ when
$f_{L_-}$ and $f_{L_+}$ are nonzero. Suppose
$\8{deg}(\xi_+)\le\8{deg}(\xi_-)$ and let $r_0=\xi_+$ and
$r_1=\xi_-$. We make a division with remainder and get a polynomial
$q_0$ with $\8{deg}(q_0)=\8{deg}(r_0)-\8{deg}(r_1)$ and a polynomial
$r_2$ with $\8{deg}(r_2)<\8{deg}(r_1)$ such that
\begin{equation}
r_0=q_0 r_1 + r_2.
\end{equation}
With this decomposition we get $\8{gcd}(r_0,r_1)=\8{gcd}(r_1,r_2)$.
We repeat this division recursively until the remainder vanishes:
\begin{eqnarray}
r_i&=&q_i r_{i+1}+r_{i+2}\\
r_{n+1}&=&q_{n+1}r_n.
\end{eqnarray}
Then we have $r_n=\8{gcd}(r_n,r_{n+1})=\8{gcd}(r_0,r_1)$. We rewrite
the recursion to get the form of equation (\ref{euclid}):
\[\left(\begin{array}{c}r_{i-1}\\r_i\end{array}\right)=\left(\begin{array}{cc}0&1\\1&-q_{i-2}\end{array}\right)\left(\begin{array}{c}r_{i-2}\\r_{i-1}\end{array}\right)\]
So we get
\[\left(\begin{array}{c}r_n\\0\end{array}\right)=\Gamma_n\dots\Gamma_0\left(\begin{array}{c}r_0\\r_1\end{array}\right)\]
with
\[\Gamma_i:=\left(\begin{array}{cc}0&1\\1&-q_i\end{array}\right)\;,\]
and since all entries in the matrices are polynomials we get
polynomials $f_0$ and $f_1$ such that
\[r_n=f_0 r_0+f_1 r_1\]
holds.
\end{proof}
%%%

%%%
%%%
\section{Main results}
%%%
%%%
\label{sec-3}

%%%
\subsection{Characterization of Clifford quantum cellular automata}
%%%

We have seen in Proposition~\ref{sca-fourier} that symplectic
cellular automata are nothing else but linear functions on the phase
space $\glph{s}=\fring{s}^2$ that preserve the
$\fring{s}$-symplectic form $\FMspf$. Such a map on $\fring{s}^2$
can be represented by a two-by-two matrix with entries in the
polynomial ring $\fring{s}$. The first column is given by $\1t_1=\1t
{1\choose 0}$ (``the local rule for $X$'') and the second column by
$\1t_2=\1t {0\choose 1}$ (``the local rule for $Z$''). The
commutation relations of the local rule then end up in the following
conditions on the column vectors:

%%%
\begin{cor}\label{cor-sca}
A two-by-two matrix $\1t$ with entries in $\fring{s}$ is a
symplectic cellular automaton, if and only if, the column vectors of
$\,\1t=(\1t_1,\1t_2)$ fulfill
$\fmspf{\1t_1}{\1t_1}=0=\fmspf{\1t_2}{\1t_2}$ and
$\fmspf{\1t_1}{\1t_2}=1$.
\end{cor}
%%%

%%%
\begin{remark}
\rm The column vectors $\1t_{1,2}$ of a symplectic cellular
automaton generate maximally isotropic $\fring{s}$-subspaces
$\fring{s}\1t_{1,2}$, since these are the images of the basis
vectors ${1\choose 0}$ and ${0\choose 1}$ under the invertible
symplectic transformation $\1t$. Because the basis vectors generate
by construction maximally isotropic subspaces, this must then also
be true for the images $\1t_{1,2}$.\fin
\end{remark}
%%%

In the next subsection, we shall see that the classification of
one-dimensional symplectic cellular automata is easier to handle. A
useful observation is that a $s$-dimensional symplectic cellular
automaton $\1t\in\2M_2(\fring{s})$ induces for each direction
$k=1,\dots, s$ a one-dimensional cellular automaton. To see this, we
introduce for each direction $k=1,\dots, s$  a surjective ring
homomorphism $\oned{k}{}$ which maps the polynomial ring $\fring{s}$
of $s$-variables $\Vari{1}{s}$ onto the ring $\fring{}$ of one
variable $\vari$. The ring homomorphism $\oned{k}{}$ assigns to a
polynomial $f\in\fring{s}$ the polynomial
\begin{equation}
\label{eq-red-dim}
\oned{k}{f}(\vari):=\sum_{(x^1,\dots,x^s)\in\7Z^s} f(x^1,\dots,x^s)
\ \vari^{x^k}
\end{equation}
which only depends on the variables $\vari,\vari^{-1}$. The ring
homomorphism $\oned{k}{}$ evaluates the polynomial $f\in\fring{s}$
at $\vari_l=1$, for $l\not=k$, whereas $\vari_k=\vari$ is the
remaining free variable.

For a symplectic cellular automaton $\1t\in\2M_2(\fring{s})$ the
conditions $\fmspf{\1t_{1,2}}{\1t_{1,2}}=0$ and
$\fmspf{\1t_1}{\1t_2}=1$ are identities of polynomials. The matrix
$\oned{k}{\1t}\in\2M_2(\fring{})$ is build by applying the ring
homomorphism $\oned{k}{}$ to each matrix element individually.
Obviously, the identities
$\oned{k}{\fmspf{\1t_{1,2}}{\1t_{1,2}}}=\fmspf{\oned{k}{\1t}_{1,2}}{\oned{k}{\1t}_{1,2}}=0$
as well as
$\oned{k}{\fmspf{\1t_1}{\1t_2}}=\fmspf{\oned{k}{\1t}_1}{\oned{k}{\1t}_2}=1$
follow. As a consequence we get:

%%%
\begin{cor}\label{cor-red-dim}
Let $\1t\in\2M_2(\fring{s})$ be a $s$-dimensional symplectic
cellular automaton. Then for each direction $k=1,\dots, s$, the
two-by-two matrix $\oned{k}{\1t}\in\2M_2(\fring{})$ is a
one-dimensional symplectic cellular automaton.
\end{cor}
%%%

Now it is easy to show that symplectic cellular automata are
reflection invariant and that the determinant is a monomial. It is
slightly more involved that we have reflection invariance with
respect to a lattice point and not with respect to an half integer
lattice point.

%%%
\begin{thm}\label{thm-sca-1}
A $\fring{s}$-linear map $\1t\in\2M_2(\fring{s})$ is a symplectic
cellular automaton, if and only if, the following holds:
\begin{enumerate}
\item
The matrix $\1t$ is a reflection invariant with respect to some
lattice point $a\in\7Z^s$.
\item
The $\fring{s}$-valued determinant of $\1t$ is
$\8{det}(\1t)=\vari^{2a}$.
\end{enumerate}
\end{thm}
%%%
\begin{proof}
If $\1t$ is a symplectic cellular automaton, then the column vectors
$\1t_{1,2}$ generate maximally isotropic subspaces. By
Lemma~\ref{lem-iso-1} it follows that $\1t_1$, respectively $\1t_2$,
is reflection invariant to some half integer lattice point $a$,
respectively $b$. Since $\1t$ preserves the symplectic form $\FMspf$
we obtain
$1=\fmspf{\1t_1}{\1t_2}=\fmspf{\vari^{2a}\radj{{\1t}_1}}{\vari^{2b}\radj{{\1t}_2}}=\vari^{2(b-a)}
\fmspf{\radj{{\1t}_1}}{\radj{{\1t}_2}}=\vari^{2(b-a)}$ and therefore
$a=b$ for an half-integer lattice point $a\in\frac{1}{2}\7Z^s$. As a
consequence, $\1t$ is reflection invariant for
$a\in\frac{1}{2}\7Z^s$. Now,
$1=\fmspf{\1t_1}{\1t_2}=\fmspf{\vari^{2a}\radj{{\1t}_1}}{\1t_2}=\vari^{-2a}\det(\1t_1,\1t_2)$.

Vice versa, let $\1t$ be a matrix, which is invariant with respect
to the reflection at $a$ and whose determinant is
$\8{det}(\1t)=\vari^{2a}$. Then the column vectors $\1t_{1,2}$ are
reflection invariant, which implies (by Lemma~\ref{lem-iso-1}) that
$\fmspf{\1t_{1,2}}{\1t_{1,2}}=0$ holds. The determinant of $\1t$ is
$\vari^{2a}$ which implies $\fmspf{\1t_1}{\1t_2}=1$. Thus $\1t$
preserves the symplectic form $\FMspf$.

By Corollary~\ref{cor-red-dim}, we obtain a one-dimensional
symplectic cellular automaton $\oned{k}{\1t}$ for each lattice
direction $k=1,\dots, s$. We have already shown that the column
vectors $\1t_{1,2}$ are reflection invariant for
$a=(a^1,\dots,a^s)$, which implies that for each direction $k$ the
column vectors $\oned{k}{\1t}_{1,2}$ are reflection invariant for
$a^k$. We also have that $\oned{k}{\1t}_{1,2}$ generate maximally
isotropic $\fring{}$-subspaces, since these define valid cellular
automaton rules.

Suppose now, that $f\in\fring{}$ is reflection invariant for
$b\in\frac{1}{2}\7Z$ in the half-integer lattice. Then we can
translate $f$ by an even translation $2y\in 2\7Z$, such that
$c=2(b+y)$ is either $0$ or $1$. If $f$ is of even length, then
$c=1$ follows. The polynomial $\vari^yf$ is reflection invariant for
$1/2$ and can be expanded as
\begin{equation}
\vari^yf=\sum_{n\in\7N} c_n \ (\vari^{n+1}+\vari^{-n})  \; .
\end{equation}
Now, for each $n\in\7N$, the polynomial $\vari^{n+1}+\vari^{-n}$ is
a multiple of $u+1$. Thus  $f$ is also a multiple\footnote{Note that
the coefficients are from the finite field $\Field{p}$.} of
$\vari^{-y}(\vari+1)$. From this we conclude that, if $b$ is not an
integer, then a reflection invariant $\xi\in\fring{}^2$ is a
multiple of $\vari^{-y}(\vari+1)$ and does not generate a maximally
isotropic $\fring{}$-subspace, since $\vari+1$ is a nontrivial
common divisor of $\xi_+$ and $\xi_-$, which is a contradiction. So
$a^k$ must be an integer lattice point, that is,
$a=(a^1,\dots,a^s)\in\7Z^s$.
\end{proof}
%%%

So each symplectic cellular automaton $\1t$ is reflection invariant
for the reflection at some lattice point $a\in\7Z^s$. Therefore, the
symplectic cellular automaton $\vari^{-a}\1t$ is reflection
invariant with respect to the origin $x=0$. In the subsequent, we
call all symplectic cellular automata, which are reflection
invariant with respect to the origin, to be ``centered'' and it is
sufficient to classify only those. The polynomials in $\fring{s}$
which are reflection invariant with respect to the origin form a
subring $\subfring{s}\subset\fring{s}$ and will be simply called
reflection invariant (for $s=1$ we will again omit the index). From
Theorem~\ref{thm-sca-1}, we obtain a handy characterization of
centered symplectic cellular automata:

%%%
\begin{cor}
\label{cor-sca-1} The group of centered symplectic cellular automata
is given by the group $\SL{2}{\subfring{s}}$ of two-by-two matrices
$\1t$ with entries in the subring $\subfring{s}$ of reflection
invariant polynomials and $\subfring{s}$-valued determinant
$\8{det}(\1t)=1$.
\end{cor}
%%%

%%%
\begin{example}\rm
The symplectic transformation corresponding to the ``cluster state
QCA'' (eq.~\ref{example-0}) is given by
\begin{equation}
\1t=\left(\begin{array}{cc}0&1\\1&\vari+\vari^{-1}\end{array}\right)\;.
\end{equation}
Obviously all entries are reflection invariant with respect to the
origin and the determinant is equal to one (modulo
2).\fin\end{example}
%%%

%%%
\begin{remark}\rm
A nice aspect of Corollary~\ref{cor-sca-1} is that the centered
symplectic cellular automata can be obtained by the following
strategy: Choose two arbitrary reflection invariant
$f,h\in\subfring{s}$ and find all possible factorizations of the
polynomial $fh-1=f'h'$ into a product of two reflection invariant
$f',h'\in\subfring{s}$. The corresponding symplectic cellular
automaton is then given by
\begin{equation}
\left(\begin{array}{cc}f&f'\\h'&h\end{array}\right)\in\SL{2}{\subfring{s}}
\; .
\end{equation}
Even if the task of factorizing the polynomial $fh-1$ is quite
cumbersome, there is always a ``trivial'' solution, namely, $h'=1$
and $f'=fh-1$. The matrix
\begin{equation}
\label{constr-1}
\left(\begin{array}{cc}f&fh-1\\1&h\end{array}\right)\in\SL{2}{\subfring{s}}
\end{equation}
describes the corresponding symplectic cellular automaton.
\fin
\end{remark}
%%%

%%%
\begin{remark}\rm
Another remarkable fact is that, due to Cramer's rule, the inverse
of a centered symplectic CA $\1t$ is simply given by
\begin{equation}
\1t^{-1}=\left(\begin{array}{rr}\1t_{22}&-\1t_{12}\\-\1t_{21}&\1t_{11}\end{array}\right)\;.
\end{equation}
Similarly we have that for a symplectic CA $\1t$ containing a
translation by $a$ positions, e.g $\det(\1t)=\vari^{a}$, the inverse
contains a translation by $-a$ positions.\fin
\end{remark}
%%%

%%%
\subsection{1D CQCAs and translationally invariant stabilizer states}
%%%

In this subsection we are investigating one-dimensional symplectic
cellular automata. As already mentioned, we can achieve more results
in this case, because we can apply the euclidean algorithm
(Lemma~\ref{lem-euclid}). We will use the euclidean algorithm to
show that for every reflection invariant $\xi\in\fring{}^2$ with
$\8{gcd}(\xi_+,\xi_-)=1$, there exists at least one corresponding
reflection invariant $\eta$ such that $\fmspf{\eta}{\xi}=1$ holds
and is therefore a valid column of a symplectic cellular automaton
matrix. We will use this fact to show that every uniquely determined
and translationally invariant stabilizer state can be prepared from
a product state by applying one timestep of a Clifford QCA.

Stabilizer states are studied extensively in the last years
(\cite{Got97} and \cite{NieChu00} are just examples, which are
useful as introductory texts). The basic concept is to fix an
abelian group of operators (usually a subgroup of the Pauli group),
also called stabilizer group, and to define a stabilizer state as
common eigenvector of all these operators. In our case we are
looking for translationally invariant states, so the stabilizer
group is generated by all translates of one single Weyl operator
$\weyl{\xi}$ for some phase space vector $\xi\in\fring{}^2$. The
state should fulfill
$\omega(\Tra{x}\weyl{\xi})=\omega(\weyl{\tra{x}\xi})=1$ for all
$x\in\7Z$. The stabilizer formalism is often studied for finitely
many qudits. In that case it is known that the stabilizer state is
uniquely determined, if the number of generating operators is large
enough (see e.g. \cite{NieChu00} for a quantitative statement). In
our situation we have infinitely many qudits, so we cannot apply
this result. But it turns out that the operators $\weyl{\tra{x}\xi}$
must generate a maximal abelian algebra, or equivalently, the
subspace $\fring{}\xi$ must be maximally isotropic.

%%%
\begin{thm}\label{structure-thm-1}
For a phase space vector $\xi\in\fring{}^2$ the following is
equivalent:
\begin{enumerate}
\item\label{equiv-1}
There exists a uniquely determined state $\omega$ with
$\omega(\weyl{\tra{x}\xi})=1$ for all $x\in\7Z$.
\item\label{equiv-2}
$\fring{}\xi$ is a maximally isotropic $\fring{}$-subspace.
\item\label{equiv-3}
There is a Clifford QCA $T$ with $\weyl{\xi}=T(\weyl{0,1})$.
\item\label{equiv-4}
$\xi$ is a reflection invariant and $\8{gcd}(\xi_+,\xi_-)=1$.
\end{enumerate}
\end{thm}
%%%
\begin{proof}
\ref{equiv-2}. $\Longrightarrow$ \ref{equiv-4}. Because
$\fring{}\xi$ is a maximally isotropic subspace we conclude from
Lemma~\ref{lem-iso-1} that $\xi$ is reflection invariant with
$\8{gcd}(\xi_+,\xi_-)=1$.

\ref{equiv-4}. $\Longrightarrow$ \ref{equiv-3}. We have to find
$\eta\in\fring{}^2$ with $\fmspf{\eta}{\xi}=1$ and
$\fmspf{\eta}{\eta}=0$. With Lemma \ref{lem-euclid} we find a
solution $f_{\pm}\in\fring{}$ of the equation
$f_+\xi_++f_-\xi_-=\8{gcd}(\xi_+,\xi_-)=1$ and
$\eta':=(\radj{f}_-,-\radj{f}_+)$ is a solution of
$\fmspf{\eta'}{\xi}=1$. Yet we do not know, whether $\eta'$ is
reflection invariant, or equivalently, whether
$\fmspf{\eta'}{\eta'}=0$ holds. But if $\eta'$ is a solution of
$\fmspf{\eta'}{\xi}=1$ then the same is true for $\eta=\eta'+f\xi$.
Thus we have to solve the condition $0=\fmspf{\eta}{\eta}=
\fmspf{\eta'}{\eta'}+\fmspf{f\xi}{\eta'}+\fmspf{\eta'}{f\xi}=\fmspf{\eta'}{\eta'}-\radj{f}+f$.
The polynomial $h:=\fmspf{\eta'}{\eta'}$ is anti-symmetric with
respect to the reflection $\xi\mapsto\radj{\xi}$ and it can be
expanded as $h=\sum_{n>0} h_n(\vari^n-\vari^{-n})$. By choosing
$f=\sum_{n>0} h_n \vari^n$ we find that $\eta=\eta'+f\xi$ is indeed
reflection invariant. The matrix $(\xi,\eta)\in\SL{2}{\subfring{}}$
is then a symplectic cellular automaton and induces a Clifford QCA
with the desired property.

\ref{equiv-3}. $\Longrightarrow$ \ref{equiv-1}. Consider a state
$\omega$ with the desired property. Then this state is equal to
$\tilde{\omega}\circ T$, whereby $\tilde{\omega}$ is a state with
$\tilde{\omega}(\Tra{x}\weyl{0,1})=1$ for all $x\in\7Z$, so the
stabilizer group of this state is given by all translates of
$\weyl{0,1}$. This means $\tilde{\omega}$ is a translationally
invariant product state, which is determined by the equation
$\tilde{\omega}(\weyl{0,1})=1$ and corresponds to the one
dimensional projector onto the eigenspace of $\weyl{0,1}$ with
eigenvalue $1$. Therefore this state is uniquely determined and
$\omega=\tilde{\omega}\circ T$ is the unique state with
$\omega(\weyl{\tra{x}\xi})=1$.

\ref{equiv-1}. $\Longrightarrow$ \ref{equiv-2}. Suppose
$\fring{}\xi$ is an isotropic $\fring{}$-subspace but not maximally
isotropic. By Lemma \ref{iso-inbedding}, we know that there exists a
phase space vector $\xi'$ with $\fring{}\xi\subsetneqq\fring{}\xi'$.
So we have $\xi=f\xi'$ with $f$ not invertible and Lemma
\ref{lem-iso-1} tells us that $\xi'$ is reflection invariant. With
help of the euclidean algorithm we find a QCA $T$ and a
corresponding symplectic transformation $\1t$ with
$T(\weyl{0,1})=\weyl{\1t(0,1)}=\weyl{\xi'}$ (just as step two of
this proof). Now consider a product state $\varphi$ with
$\varphi(\weyl{0,\vari^x})=\exp(\frac{2\pi i}{p}a_x)$ depending on
the $a_x$. We transform this state with $T^{-1}$ and the expectation
values of the operators $\tra{x}\weyl{\xi}$ should be all equal to
1:
\begin{eqnarray*}
1&=&\varphi_{T^{-1}}(\tra{x}\weyl{\xi})=\varphi(\tra{x}\weyl{\1t^{-1}\xi})\\
&=&\varphi(\tra{x}\weyl{\1t^{-1}(f\xi')})=\varphi(\tra{x}\weyl{\hat{f}\convol\widehat{
(0,1)}})\\
&=&\varphi\Big(\weyl{\sum_k
\hat{f}_{-k}\delta_{k+x}(0,1)}\Big)=\prod_k
\hat{f}_{-k}\varphi(\weyl{\delta_{k+x}(0,1)})\\
&=&\exp\Big(\frac{2\pi i}{p}\sum_k \hat{f}_{-k}a_{k+x}\Big)
\end{eqnarray*}
So we have to solve the equations $\sum_k \hat{f}_{-k}a_{k+x}=0$ to
get appropriate $a_x$ and therefore states with the desired
property. Since $f$ is not invertible the support of $\hat{f}$ is
not a one-elementary set. Let $I=\{-L_-,\dots,-L_+\}$ be the minimal
interval such that $\supp{\hat{f}}\subset I$. We can choose
arbitrary $a_{L_-},\dots,a_{L_+-1}$ to compute $a_{L_+}$ from the
equation $\sum_k \hat{f}_{-k}a_k$. Recursively all $a_x$ can be
calculated from the other equations but the solution will depend
from the initial choice of the $a_{L_-},\dots,a_{L_+-1}$. This means
that there exists more than one state $\varphi$ of the above form,
such that $1=\varphi_{T^{-1}}(\tra{x}\weyl{\xi})$ is fulfilled. So
the uniqueness of the state in \ref{equiv-1}. forces $\fring{}\xi$
to be maximally isotropic.
\end{proof}
%%%

So we have shown that every translationally invariant and uniquely
determined stabilizer state in a one-dimensional lattice can be
prepared out of a product state by a single timestep of a Clifford
QCA. Unfortunately we cannot generalize this result to higher
lattice dimensions, because Lemma \ref{lem-euclid} is only valid for
univariate polynomials. The euclidean algorithm for computing the
greatest common divisor can be generalized to multivariate
polynomials \cite{Bro71}, but the extended version (equation
(\ref{euclid})) does not hold.

%%%
\begin{example}
\rm We consider again the phase space vectors
$\xi_1=(1+\vari){0\choose 1}$ and $\xi_2={1\choose
\vari+\vari^{-1}}$ (see Example~\ref{stand-ex}). As already
mentioned, the phase space vector $\xi_1$ is reflection invariant
for $1/2$ and generates an isotropic $\fring{}$-subspace, but none
of the statements of Theorem~\ref{structure-thm-1} holds: The
expectation value of $\weyl{\tra{x}\xi_1}=Z_x\otimes Z_{x+1}$ is
equal to one both in the ``all spins up'' and in the ``all spins
down'' state, so there is no uniquely determined stabilizer state.
As we have seen in~\ref{stand-ex} the subspace $\fring{}\xi_1$ is
not maximally isotropic. The reflection invariance does not hold for
an integer lattice point, so $\xi_1$ is not a valid column of a
symplectic CA, and $1+\vari$ is a common divisor of $\xi_{1,+}$ and
$\xi_{1,-}$, which is not invertible.

In contrast $\xi_2$ fulfills all four conditions. The uniquely
determined stabilizer state is given by the one-dimensional cluster
state and a possible CQCA is given by example~\ref{example-0}.
\fin
\end{example}
%%%

%%%
\subsection{Factorization of 1D Clifford QCAs}
%%%
We have seen that the set of centered CQCAs form a group and that
this group is given by $2\times 2$-matrices with determinant one and
reflection invariant polynomials as matrix elements. In the
one-dimensional case the group structure can be more clarified,
since we are able to give a complete set of generators, which can be
regarded as elementary operations.

A simple example of a $2\times 2$-matrix in $\SL{2}{\subfring{s}}$
is for some reflection invariant polynomial $f\in\subfring{}$ given
by
\begin{equation}
\gen{f}:=\shear{f} \; ,
\end{equation}
which we will call ``shear transformation''. In particular,
$\gen{f_1+f_2}=\gen{f_1}\gen{f_2}$ holds for all
$f_1,f_2\in\subfring{}$\footnote{This means, the map
$f\mapsto\gen{f}$ is a group homomorphism from the additive group
$\subfring{}$ into the group of centered symplectic cellular
automata $\SL{2}{\subfring{}}$.}. The symmetric polynomials
$\symvari{n}=\vari^n+\vari^{-n}$, $n\in\7N$, and $\symvari{0}=1$
form a basis of the subring $\subfring{}$. Thus every shear
transformation can be decomposed into a finite product of elementary
shear transformations $\gen{c\symvari{n}}$ with $n\in\7N\cup\{0\}$
and $c\in\field$.

The local rule of the corresponding QCA $G_n$ with
$G_n\weyl{\xi}=\weyl{\gen{\symvari{n}}\xi}$ is for $n\ge 1$ given by
\begin{equation}
\begin{array}{lclcr}
G_n(X_0)&=&Z_{-n}\otimes\I\otimes\cdots\otimes
&X_0&\otimes\I\otimes\cdots\otimes Z_n\\
G_n(Z_0)&=& &Z_0&\end{array}\;.
\end{equation}
For $n=0$ we have the single cell operation (``local shear
transformation'')
\begin{equation}
\begin{array}{lcc}
G_0(X_0)&=&\weyl{1,1}\\
G_0(Z_0)&=&Z_0
\end{array}\;,
\end{equation}
which correspond for $p=2$ to applying the phase gate
$\left(\begin{array}{cc}1&\\&i\end{array}\right)$ to all single
cells.

Another single cell operation is the ``local Fourier
transformation'', which is in phase space given by the matrix
\begin{equation}
\genp{c}=\ftra{c}
\end{equation}
with some constant $0\ne c\in\field$ (for $c=1$ we will write
$\genf:=\genp{1}$). For $p=2$ we have $c=1$ and the corresponding
QCA switches the operators $X$ and $Z$ in each single cell and is
therefore given by applying the Hadamard matrix.

Note that all symplectic single cell transformations can be obtained
by a product of local shear and local Fourier transformations, which
is a generalization to higher cell dimensions of the fact that local
Clifford operations are generated by Hadamard and phase gate.

The symplectic transformations $\gen{f},\genp{c}$ are elementary
symplectic cellular automata in the sense of the following theorem.
The proof, which is technically slightly more involved, is given in
the Appendix~\ref{app-3}.

%%%
\begin{thm}
\label{structure-thm-2} The group of centered symplectic cellular
automata $\SL{2}{\subfring{}}$ is generated by the set
$\{\gen{\symvari{n}}|n\in\7N\cup\{0\}\}\cup\{\genp{c}|c\in\field\}$.
\end{thm}
%%%

\begin{remark}
\rm A more concrete formulation of the statement of
Theorem~\ref{structure-thm-2} is that every one-dimensional centered
symplectic cellular automaton $\1t$ is a finite product of shear
transformations and local fourier transforms of the following form:
\begin{equation}
\label{decomp}
\1t=\gen{f_r}\genp{c_r}\cdots\gen{f_2}\genp{c_2}\gen{f_1}\genp{c_1}
\end{equation}
with reflection invariant polynomials $f_1,\dots, f_r\in\subfring{}$
and constants $c_1,\dots,c_r\in\field$. \fin
\end{remark}
%%%

%%%
\begin{example}
\rm Let us consider in the qubit case ($p=2$) the symplectic
cellular automaton
\begin{equation}\label{exa-sca}
\1t=\left(\begin{array}{cc}\symvari{1}&1\\1+\symvari{2}&\symvari{1}\end{array}\right)
\end{equation}
and, since $\symvari{2}=\symvari{1}^2$ for $p=2$, we have
$\det(\1t)=1$. The corresponding CQCA is given by
\begin{equation}
\begin{array}{lcrcl}
T(X_0)&=&Z_{-2}\otimes X_{-1}\otimes &Z_0&\otimes X_1\otimes Z_2\\
T(Z_0)&=&Z_{-1}\otimes &X_0&\otimes Z_1
\end{array} \; .
\end{equation}
The basic idea for deriving a decomposition like (\ref{decomp}) is
to reduce the support of the first column of $\1t$ by applying a
shear transformation from the right. We get
\begin{equation}
\1t\gen{\symvari{1}}=\left(\begin{array}{cc}0&1\\1&\symvari{1}\end{array}\right)\;.
\end{equation}
This matrix is obviously equal to $\gen{\symvari{1}}\genf$ and we
have
\begin{equation}
\1t=\gen{\symvari{1}}\genp{c}\gen{\symvari{1}}\;,
\end{equation}
which is indeed a decomposition in accordance with (\ref{decomp}).
\fin
\end{example}
%%%

%%%
\begin{remark}\rm
For $p=2$ all the generators $\gen{f}$ and $\genf$ are their own
inverses, so the time evolution of these operations alternates
between the identity and a single timestep of the automaton.
Especially these QCAs show no propagation, because the neighborhood
of the iterated automaton does not increase with the number of
timesteps. A nontrivial time evolution only occurs, if the
symplectic cellular automaton is composed of at least two different
generators. \fin\end{remark}
%%%

%%%
%%%
\section{Periodic boundary conditions}
%%%
%%%
\label{sec-periodic} In this chapter we are looking for
translationally invariant Clifford operations with periodic boundary
conditions on an $s$-dimensional lattice. These boundary conditions
are given by an $s$-dimensional torus $\7T_N^s$, which is determined
by $s$ independent lattice vectors $N=(N_1,\dots,N_s)$ (see
figure~\ref{torus}), and all lattice points which differ by these
vectors are identified. The number of (not identified) lattice
points is given by $|\,\7T_N^s|:=|\det(N_1,\dots,N_s)|$. We denote
here by $\fring{s,N}$ the ring of polynomials
$f=\sum_{x\in\7T_N^s}f(x) u^x$ such that the $u$-variables fulfill
the periodic boundary conditions $u^{N_1}=u^{N_2}=\dots =u^{N_s}=1$.
This guaranties that the product of two polynomials from
$\fring{s,N}$ is again an element from $\fring{s,N}$. But
algebraically there are large differences between $\fring{s}$ and
$\fring{s,N}$: $\fring{s,N}$ is not a division ring, because there
are zero divisors and there are in general other invertible elements
than $u^x=u_1^{x_1}u_2^{x_2}\cdots u_s^{x_s}$. The reflection
$\radj{f}$ is again given by replacing $\vari$ by $\vari^{-1}$ or in
other words we substitute $\vari^x$ by $\vari^{N-x}$. The symplectic
form $\FMspf$ is then of the same form as in the infinite lattice
case.

%%%
\begin{figure}[ht]

\centerline{\includegraphics[scale=0.4]{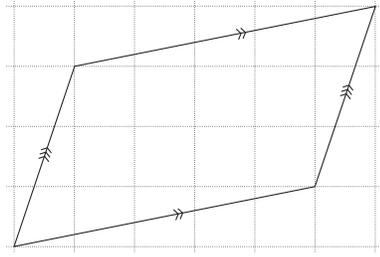}}

\caption{\label{torus} A $2$-dimensional torus defined by
$N_1=(1,3)$ and $N_2=(5,1)$.}
\end{figure}
%%%

Now we have to say, how a Clifford QCA (or a symplectic cellular
automaton) is defined on these systems. In the general theory of
QCAs~\cite{SchuWer04}, the neighborhood of a QCA with periodic
boundary conditions is not allowed to be too large in comparison
with the torus. This guaranties that the QCA can be extended to the
whole lattice. Since this case is covered by restricting the
existing Clifford QCAs to periodic boundary conditions, we drop all
locality conditions, and we take the same structure as in Corollary
\ref{cor-sca} as definition of a symplectic cellular automaton:

%%%
\begin{defi}\em
A $2\times 2$ matrix $\1t=(\1t_1,\1t_2)$ with entries in
$\fring{s,N}$ is a symplectic cellular automaton if the column
vectors fulfill $\fmspf{\1t_1}{\1t_1}=0=\fmspf{\1t_2}{\1t_2}$ and
$\fmspf{\1t_1}{\1t_2}=1$.
\end{defi}
%%%

With this definition it is possible to state an analogous version of
Theorem~\ref{structure-thm-1} also for periodic boundary conditions.
But the proof is quite different from the infinite lattice case.
%%%
\begin{thm}
For a phase space vector $\xi\in\fring{s,N}^2$ the following is
equivalent:
\begin{enumerate}
\item\label{equiv1-1}
There exists a uniquely determined state $\omega$ with
$\omega(\weyl{\tra{x}\xi})=1$ for all $x\in\7T_N^s$.
\item\label{equiv1-2}
$\fring{s,N}\xi$ is a maximally isotropic $\fring{s,N}$-subspace.
\item\label{equiv1-3}
There is a symplectic cellular automaton $\1t$ with
$\xi=\1t{0\choose 1}$.
\end{enumerate}
\end{thm}
%%%

\begin{proof}
\ref{equiv1-1}. $\Longleftrightarrow$ \ref{equiv1-2}. A stabilizer
state on $M$ qudits is uniquely determined, if and only if, the
minimal number of generators of the stabilizer group equals $M$
\cite{NieChu00,Schl03}. Here we have the $|\,\7T_N^s|$ generators
$\weyl{\tra{x}\xi}$. These are independent, if and only if, they
generate a maximally abelian algebra, or equivalently, if
$\fring{s,N}\xi$ is a maximally isotropic subspace.

\ref{equiv1-2}. $\Longrightarrow$ \ref{equiv1-3}. Since we have a
finite dimensional space, there exists a symplectic basis, and any
basis of a maximally isotropic subspace can be extended to a
symplectic basis \cite{McDSal98}. For this construction we turn to
the original phase space (by inverse ``Fourier transform'') and
define by $\xi_x:=\tra{x}\hat{\xi}$ basis vectors of a subspace.
Since we know by~\ref{equiv1-2}. that this space is isotropic, these
vectors fulfill $\spf{\xi_x}{\xi_y}=0$ and therefore
$0=\mspf{\hat{\xi}}{\hat{\xi}}=\fmspf{\xi}{\xi}$. Then there exists
a dual vector $\hat{\eta}$ with
$\spf{\hat{\eta}}{\xi_x}=\delta_{x0}$ and we define
$\eta_y=\tra{y}\hat{\eta}$. We get that
$\spf{\eta_y}{\xi_x}=\spf{\hat{\eta}}{\xi_{x-y}}=\delta_{xy}$ holds
and so we have $\fmspf{\eta}{\xi}=1$. We have to verify that we can
choose $\eta$, such that $\fmspf{\eta}{\eta}=0$ holds. But if $\eta$
is a solution to $\fmspf{\eta}{\xi}=1$ the same is true for
$\eta'=\eta+f\xi$ for any $f\in\fring{s,N}$ and we can find similar
to the case of Theorem \ref{structure-thm-1} an appropriate $f$ with
$\fmspf{\eta'}{\eta'}=0$.

\ref{equiv1-3}. $\Longrightarrow$ \ref{equiv1-2}. Suppose that $\1t$
is a symplectic cellular automaton with $\xi=\1t{0\choose 1}$. Then
$\1t$ induces a homomorphism between the maximally isotropic
subspace $\fring{s,N}{0\choose 1}=0\oplus\fring{s,N}$ and
$\fring{s,N}\xi$ with $\1t f {0\choose 1}=\1t{0\choose f}=f\xi$.
Since $\1t$ is invertible and preserves the symplectic form
$\FMspf$, it follows that any maximally isotropic subspace is mapped
onto a maximally isotropic subspace, which implies that
$\fring{s,N}\xi$ is maximally isotropic.
\end{proof}
%%%

%%%
\begin{example}
\rm As an example of a translationally invariant stabilizer state,
we consider translationally invariant graph states \cite{Schl03}.
The graph is encoded by its adjacency matrix
$\Gamma=(\Gamma(x,y))_{x,x\in\7T_N^s}$, and the isotropic subspace
that determines the graph state is given by the phase space vectors
\begin{equation}
{\Gamma f\choose f}\in\glph{s}=\fffunct{s}^2
\end{equation}
with $f\in\fffunct{s}$. Translation invariance of the graph state
implies that the matrix elements $\Gamma(x,y)$ depend only on the
difference $x-y$, so there is a function $\gamma\in\fffunct{s}$ such
that $\Gamma(x,y)=\gamma(x-y)$ holds. Thus after Fourier transform
the phase space vector, which generates the maximally isotropic
subspace, is given by $\xi={\hat\gamma\choose 1}$ and we can choose
a suitable symplectic cellular automaton $\1t$ with
$\xi=\1t{0\choose 1}$ by
\begin{equation}
\1t=\left(
\begin{array}{cc}
1&\hat\gamma\\
0&1
\end{array}\right) \; .
\end{equation}

%%%
\begin{figure}[ht]
\centerline{\includegraphics[scale=0.6]{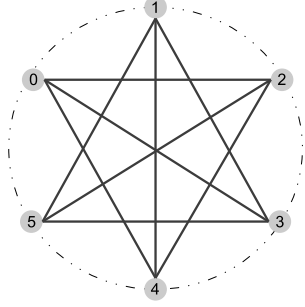}}

\caption{\label{FIG-1} A translationally invariant graph state on
the 1D torus $\7T^1_6=\7Z_6$.} \vspace{10pt} \hrule
\end{figure}
%%%

Figure~\ref{FIG-1} represents a translationally invariant graph
state on the 1D torus $\7T^1_6=\7Z_6$. The adjacency matrix $\Gamma$
is given by
\begin{equation}
\Gamma=\left(\begin{array}{cccccc}&&1&1&1&\\&&&1&1&1\\1&&&&1&1\\1&1&&&&1\\1&1&1&&&\\&1&1&1&&\end{array}\right)
\end{equation}
and obviously only depends on the difference $x-y$ of the variables
$x,y\in\7Z_6$. Applying the Fourier transform, yields the polynomial
$\hat\gamma=u^2+u^3+u^4=u^{-2}+u^2+u^3$. The symplectic cellular
automaton, which creates the graph state as explained above is given
by the matrix
\begin{equation}
\1t=\left(\begin{array}{cc}
1&u^{-2}+u^2+u^3\\
0&1
\end{array}\right) \; .
\end{equation}
Note that $\1t$ is reflection invariant, since $\radj{u^3}=u^3$ is a
reflection fix-point. \fin
\end{example}
%%%

We are going to present another example to show that the phase space
vectors do not have to be reflection invariant, because not all
invertible elements are monomials. But the invertibility of a fixed
polynomial depends on the size of the torus and so a fixed phase
space vector $\xi$ may define a translationally invariant stabilizer
state for some $N$, but it is possible that there exists $N'$, such
that $\fring{1,N'}\xi$ is not maximally isotropic, and therefore
$\xi$ does not characterize a unique stabilizer state for $N'$.

%%%
\begin{example}\rm
We consider for $p=2$ the phase space vector
\begin{equation}
\xi=(1+\vari+\vari^3){\vari^{-1}+\vari\choose 1}
\end{equation}
on a one dimensional torus of variable size. Note that
$\fmspf{\xi}{\xi}=0$ holds for all $N$, so $\xi$ generates an
isotropic subspace. The corresponding tensor product of Pauli
operators is given by
\begin{equation}
\weyl{\xi}=Z\otimes Y\otimes Y\otimes \I\otimes X\otimes Z\;,
\end{equation}
and is obviously not reflection invariant.

Let us first have a look at the case $N=7$. It is easy to show that
$1+\vari+\vari^3$ is not invertible. We define
$\tilde\xi=\xi/(1+\vari+\vari^3)={\vari^{-1}+\vari\choose 1}$ and
have that $\fmspf{\tilde\xi}{\xi}=0$, but
$\tilde\xi\notin\fring{1,7}\xi$. So $\fring{1,7}\xi$ is not
maximally isotropic, and there is no unique stabilizer state.

For $N=6$ the inverse of $1+\vari+\vari^3$ is given by
$\vari+\vari^4+\vari^5$, so $\xi$ and $\tilde\xi$ generate the same
subspace, which is actually maximally isotropic. So $\xi$ is indeed
a valid column of a symplectic automaton, but starting from the
``all spins up'' state both CQCAs corresponding to
\[
\1t=\left(\begin{array}{cc}\vari+\vari^4+\vari^5&(1+\vari+\vari^3)(\vari^{-1}+\vari)\\0&(1+\vari+\vari^3)\end{array}\right)\;,\quad\mbox{resp.}\quad\tilde{\1t}=\left(\begin{array}{cc}1&\vari^{-1}+\vari\\0&1\end{array}\right)
\]
prepare the same stabilizer state. \fin
\end{example}
%%%

%%%
%%%
\section{Conclusions}
%%%
%%%
We have analyzed the structure of Clifford quantum cellular automata
that act on a $s$-dimensional lattice of $p$-level systems. The
results which can be achieved depend on the dimension of the lattice
and whether we put periodic boundary conditions or working with the
infinite lattice.

We have characterized the group of CQCAs in terms of symplectic
cellular automata on a suitable phase space. With the help of
Fourier transform, this phase space can be identified with
two-dimensional vectors of Laurent-polynomials, and symplectic
cellular automata can be described by two-by-two matrices with
Laurent-polynomial entries. We have shown that these entries must be
reflection invariant and that up to some global shift the
determinant of the matrix must be one, so the group of CQCAs is
isomorphic to the special linear group of two-by-two matrices with
reflection invariant polynomials as matrix elements.

We have proven that there is a correspondence between 1D CQCAs and
1D translationally invariant stabilizer states. For a fixed
translationally invariant pure stabilizer state
$\omega^{\otimes\7Z}$, which is in particular a product state, every
other translationally invariant pure stabilizer state $\varphi$ can
be created by applying an appropriate CQCA $T_\varphi$ to the chosen
product state: $\varphi=\omega^{\otimes\7Z}\circ T_\varphi$.

Pure stabilizer states can be also characterized by maximally
isotropic subspaces. We have characterized the phase space vectors,
which generate maximally isotropic subspaces, namely their
components must be coprime and reflection invariant.

In the one-dimensional case we have also more clarified the group
structure of CQCAs. As we have shown, each one-dimensional CQCA can
be decomposed into a product of elementary shear automata and local
Fourier transforms, so the group of CQCAs is generated by this set
of operations.

For periodic boundary conditions the techniques from infinitely
extended lattices can be applied to a certain extend. According to
the discussion of translationally invariant stabilizer states on the
1D lattice, we have proven that there is an analogous correspondence
between CQCAs and translationally invariant stabilizer states with
periodic boundary conditions even in any lattice dimension.

%%%
\section*{Acknowledgments}
%%%
HV is supported by the ``DFG Forschergruppe 635''.

%%%
%%%
\begin{appendix}
%%%
%%%
\section{Proofs and technicalities}
%%%
%%%

%%%
\subsection{Ad Theorem~\ref{structure-thm-1}}
%%%

For the proof of Theorem~\ref{structure-thm-1} we need that a singly
generated isotropic subspace can always be embedded into a singly
generated maximally isotropic subspace:

%%%
\begin{lem}\label{iso-inbedding}
Let $0\ne\xi\in\fring{s}^2$ and $\fring{s}\xi$ be an isotropic, but
not maximally isotropic $\fring{s}$-subspace. Then there exists a
phase space vector $\eta\in\fring{s}^2$ such that
$\fring{s}\eta\supsetneqq\fring{s}\xi$ is maximally isotropic.
\end{lem}
%%%
\begin{proof}
$\fring{s}\xi$ is isotropic if and only if the equation
$0=\fmspf{\xi}{\xi}=\radj{\xi}_+\xi_--\radj{\xi}_-\xi_+$ holds. We
make a distinction of cases for this equation:
\begin{enumerate}
\item[i.]
$\xi_+=0$ (analogously $\xi_-=0$): Then
$\fring{s}\xi=\{0\}\oplus\fring{s}\xi_-$ and $\xi_-$ is not
invertible since this would force
$\fring{s}\xi=\{0\}\oplus\fring{s}$ to be maximally isotropic. So we
can set $\eta=(0,1)$.
\item[ii.]
$\xi_+=f\xi_-$ (analogously $\xi_-=f\xi_+$) with $f$ reflection
invariant: Then we have $\fring{s}\xi=\fring{s}\xi_-(f,1)$. We set
$\eta=(f,1)$ and get that $\fring{s}\eta$ is a maximally isotropic
subspace since $0=\fmspf{\eta}{\lambda}=f\lambda_--\lambda_+$
implies $\lambda=\lambda_-\eta\in\fring{s}\eta$.
\item[iii.]
$\xi_+\ne 0\ne\xi_-$ and $\xi_\pm\ne f\xi_\mp$: Then
$\xi=f\radj{\xi}$ with $f$ invertible, so $\xi$ is reflection
invariant for some $n\in\frac{1}{2}\7Z^d$. Because $\fring{s}\xi$ is
not maximally isotropic we can find $\eta\notin\fring{s}\xi$ with
$0=\fmspf{\xi}{\eta}=u^n(\xi_+\eta_--\xi_-\eta_+)$. Since $\xi_+$
and $\xi_-$ are nonvanishing this implies $\xi=g\eta$ for some
$g\in\fring{s}$. We can choose $\8{gcd}(\eta_+,\eta_-)=1$ and
$\fring{s}\eta$ to be maximally isotropic.
\end{enumerate}
\end{proof}
%%%

%%%
%%%
\subsection{Ad Theorem~\ref{structure-thm-2}}
%%%
\label{app-3}
%%%
For a polynomial $f\in\fring{}$ the coefficient of the monomial
$\vari^x$ is $\coef{x}{f}$. Recall that ``degree''of a polynomial in
$f\in\fring{}$ is defined by
$\deg{f}:=\max\{x|\coef{x}{f}\not=0\}-\min\{x|\coef{x}{f}\not=0\}$
and that the support is defined by
$\supp{f}:=\{x|\coef{x}{f}\not=0\}$.

%%%
\begin{lem}\label{lem-reduce}
Let $(\xi,\eta)$ be a symplectic cellular automaton which is
invariant under the reflection at the origin: $\xi=\radj{\xi}$ and
$\eta=\radj{\eta}$. If the degrees of column vectors fulfill
$\deg{\xi}>\deg{\eta}$ then there exists a shear transformation
$\gen{f}$, with reflection invariant $f\in\fring{}$, such that the
symplectic transformation
\begin{equation}
(\xi',\eta')=(\xi,\eta)\gen{f}\genp{1}
\end{equation}
satisfies $\deg{\xi,\eta}>\deg{\xi',\eta'}$ and $\deg{\xi'}>\deg{\eta'}$.
\end{lem}
%%%
\begin{proof}
Since $\xi$ and $\eta$ are reflection invariant, the degree is an
even integer and we introduce $x:=\deg{\xi}/2$, $y:=\deg{\eta}/2$,
as well as $n_1:=x-y>0$. We conclude from the identity
$\fmspf{\xi}{\eta}=1$ that
\begin{equation}
\coef{x+y}{\fmspf{\xi}{\eta}} =\coef{-x}{\xi_+}\coef{y}{\eta_-}-\coef{-x}{\xi_-}\coef{y}{\eta_+}=0
\end{equation}
is valid. This implies that
\begin{equation}
\coef{x}{\xi}=\coef{-x}{\xi}=-f_1\coef{y}{\eta}=-f_1\coef{-y}{\eta}
\end{equation}
for some $f_1\in\7F$. Now $\coef{\pm x}{\xi+\pal{f_1}{n_1}\eta}=0$ which implies that
\begin{equation}
\deg{\xi+\pal{f_1}{n_1}\eta}<\deg{\xi} \; .
\end{equation}
Now we observe
\begin{eqnarray}
(\xi_1,\eta_1)&:=&(\xi,\eta)\gen{\pal{f_1}{n_1}}
\nonumber\\
&=&
\left(\begin{array}{cc}\xi_+&\eta_+\\
\xi_-&\eta_-\end{array}\right)\shear{\pal{f_1}{n_1}}
\\
\nonumber
&=&
\left(\begin{array}{cc}\xi_++\pal{f_1}{n_1}\eta_+&\eta_+\\
\xi_-+\pal{f_1}{n_1}\eta_-&\eta_-\end{array}\right)
\end{eqnarray}
from which we conclude that $\deg{\xi_1,\eta_1}<\deg{\xi,\eta}$. If
$\deg{\xi_1}>\deg{\eta_1}$ we can find a shear transformation
$\gen{\pal{f_2}{n_2}}$ such that
\begin{equation}
(\xi_2,\eta_2)=(\xi_1,\eta_1)\gen{\pal{f_2}{n_2}}
\end{equation}
fulfills $\deg{\xi_2,\eta_2}<\deg{\xi_1,\eta_1}$. We can proceed
this reduction until the $l$th step with
$2l=|\supp{\xi}\setminus\supp{\eta}|$. The resulting symplectic
cellular automaton
\begin{equation}
(\xi_l,\eta_l)=(\xi_{l-1},\eta_{l-1})\gen{\pal{f_l}{n_l}}
\end{equation}
then satisfies $\deg{\xi_l}\leq\deg{\eta_l}$. If
$\deg{\xi_l}=\deg{\eta_l}$ then, there is an appropriate constant
$f_{l+1}\in\7F$ such that
\begin{equation}
(\xi',\eta')=(\xi,\eta)\gen{f}\genp{1}=(-\eta_{l},\xi_{l}+f_{l+1}\eta_l)
\end{equation}
holds with $\deg{\xi'}>\deg{\eta'}$. Here $f$ is the reflection
invariant polynomial
\begin{equation}
f=\sum_{j=1}^{l+1}\pal{f_j}{n_j} \; .
\end{equation}
If $\deg{\xi_l}<\deg{\eta_l}$, then the shear transformation
$\gen{f_{l+1}}$ is not applied and we get
\begin{equation}
(\xi',\eta')=(\xi,\eta)\gen{f}\genp{1}=(-\eta_{l},\xi_{l})
\end{equation}
with the polynomial $f=\sum_{j=1}^{l}\pal{f_j}{n_j}$.
\end{proof}
%%%

%%%
\begin{proof}[Proof of Theorem~\ref{structure-thm-2}.]
Let $(\xi_0,\eta_0)$ be a symplectic cellular automaton which is
invariant under the reflection at the origin. Then, by
Lemma~\ref{lem-reduce}, there exists a symplectic cellular automaton
$(\xi_1,\eta_1)$ and a shear transformation $\gen{f_1}$ such that
\begin{equation}
(\xi_0,\eta_0)=(\xi_1,\eta_1)\gen{f_1}\genp{1}
\end{equation}
and $\deg{\xi_1}>\deg{\eta_1}$. Thus we can iterate this reduction
process until $(\xi_k,\eta_k)$ is a constant symplectic
transformation (corresponding to a QCA with single cell
neighborhood), which can be decomposed into a product of local shear
transformations $\gen{\symvari{0}}$ and local Fourier transforms
$\genp{c_i}$ with $c_i\in\field$. This yields the following
decomposition of $(\xi_0,\eta_0)$:
\begin{equation}\label{decomp-app}
(\xi_0,\eta_0)=\gen{f_r}\genp{c_r}\cdots
\gen{f_2}\genp{c_2}\gen{f_1}\genp{c_1} \; .
\end{equation}
\end{proof}
%%%
%%%
\end{appendix}
%%%
%%%
%%%

%%%
%%%
%%%
\end{document}